\documentclass[a4paper,thmsb,11pt]{article}
\usepackage{color}
\usepackage{geometry}
\usepackage{geometry}
\usepackage{amsmath}
\usepackage{amssymb}
\usepackage{times}
\usepackage{mathrsfs}
\usepackage{graphicx}
\usepackage{animate}
\usepackage{algorithm}
\usepackage{algorithmicx}
\usepackage{algpseudocode}
\usepackage{upgreek}
\usepackage{tikz}
\usetikzlibrary{chains, positioning}
\usepackage[colorlinks,linkcolor=blue]{hyperref}
\newtheorem{theorem}{Theorem}

\newtheorem{problem}{Problem}
\newtheorem{definition}{Definition}

\newtheorem{lemma}{Lemma}
\newtheorem{assumption}{Assumption}
\newtheorem{proposition}{Proposition}

\begin{document}
%
\title{\textbf{Reachability Analysis and Safety Verification for Neural Network Control Systems}}

\author{Weiming~Xiang \footnotemark[1]~~and~~Taylor T. Johnson \footnotemark[1]}

\renewcommand{\thefootnote}{\fnsymbol{footnote}}
	
\footnotetext[1]{Authors are with the Department of Electrical Engineering and Computer Science, Vanderbilt University, Nashville, Tennessee 		37212, USA. Email: xiangwming@gmail.com (Weiming Xiang);  taylor.johnson@gmail.com (Taylor T. Johnson).}

\maketitle

\begin{abstract}
\boldmath
Autonomous cyber-physical systems (CPS) rely on the correct operation of numerous components, with state-of-the-art methods relying on machine learning (ML) and artificial intelligence (AI) components in various stages of sensing and control. This paper develops methods for estimating the reachable set and verifying safety properties of dynamical systems under control of neural network-based controllers that may be implemented in embedded software. The neural network controllers we consider are feedforward neural networks called multilayer perceptrons (MLP) with general activation functions. As such feedforward networks are memoryless, they may be abstractly represented as mathematical functions, and the reachability analysis of the network amounts to range (image) estimation of this function provided a set of inputs. By discretizing the input set of the MLP into a finite number of hyper-rectangular cells, our approach develops a linear programming (LP) based algorithm for over-approximating the output set of the MLP with its input set as a union of hyper-rectangular cells. Combining the over-approximation for the output set of an MLP based controller and reachable set computation routines for ordinary difference/differential equation (ODE) models, an algorithm is developed to estimate the reachable set of the closed-loop system. Finally, safety verification for neural network control systems can be performed by checking the existence of intersections between the estimated reachable set and unsafe regions. The approach is implemented in a computational software prototype and evaluated on numerical examples.

\end{abstract}

\section{Introduction}
In recent decades, there have been considerable research activities in the use of neural networks for control of complex systems such as stabilizing neural network controllers \cite{wu2014exponential,yu1998stable} and adaptive neural network controllers \cite{ge1999adaptive,hunt1992neural}. More recently, neural networks have been deployed in high-assurance systems such as self-driving vehicles \cite{bojarski2016end}, autonomous systems \cite{julian2017neural}, and aircraft collision avoidance \cite{julian2016policy}. Neural network based controllers have been demonstrated to be effective at controlling complex systems.

However, such controllers are confined to systems that comply with the lowest levels of safety integrity, since the majority of neural networks are viewed as black boxes lacking effective methods to predict all outputs and assure safety specifications for closed-loop systems. In a variety of applications to feedback control systems, there are safety-oriented restrictions that the system states are not allowed to violate while under the control of neural network based feedback controllers. Unfortunately, it has been observed that  even well-trained neural networks are sometimes sensitive to input perturbations and might react in unexpected and incorrect ways to even slight perturbations of their inputs \cite{szegedy2013intriguing}, thus it could result in unsafe closed-loop systems. With the progress of adversarial machine learning, such matters may only become worse in safety-critical cyber-physical systems (CPS). Hence, methods that are able to provide formal guarantees are in great demand for verifying specifications or properties of systems involving neural network components, especially as such AI/ML components are integrated in safety-critical CPS. The verification of neural networks is a computationally hard problem. Even verifications of simple properties concerning neural networks have been demonstrated to be non-deterministic  polynomial (NP) complete problems as reported in \cite{katz2017reluplex}.

A few results have been reported in the literature for verifying specifications in systems consisting of neural networks. In \cite{huang2017safety} satisfiability modulo theories (SMT) solvers are developed and utilized for the verification of feed-forward multi-layer neural networks. In \cite{pulina2010abstraction,pulina2012challenging} an abstraction-refinement approach is developed for computing output reachable set of neural networks. A simulation-based approach is developed in \cite{xiang2017output}, which turns the problem of over-approximating the output set of a neural network into a problem of computing the maximal sensitivity of the neural network, which is formulated in terms of a sequence of convex optimization problems. In particular, for a class of neural networks with a specific type of activation functions called rectified linear unit (ReLU), several methods are developed such as mixed-integer linear programming (MILP) based approach \cite{lomuscio2017an_arxiv,dutta2017output,dutta2018output}, linear programming (LP) based approach \cite{ehlers2017formal}, Reluplex algorithm stemmed from Simplex algorithm \cite{katz2017reluplex}, and polytope manipulation based approach \cite{xiang2017reachable_arxiv}. 

In this paper, we study the problems of reachable set estimation and safety verification for dynamical systems equipped with neural network controllers, where the plant is represented in its typical modeling formalism of time-sampled ordinary differential equations (ODEs) and the controller is a feedforward neural network determining actuation. The neural network controller considered in this paper is in the form of feedforward neural networks. The key step to solve this challenging problem is, as we shall see, to soundly estimate the output set of a feedforward neural network. By using a finite number of hyper-rectangular sets to over-approximate the input set, the output set estimation problem can be transformed into a linear programming (LP) problem for neural networks with most of classes of activation functions. Then, as the estimated output set of the neural network controller is the input set to the plant described by an ordinary difference/differential equation (ODE), the reachable set estimation for the closed-loop system can be obtained by employing sophisticated reachability analysis methods for ODE models. 

\subsection{Related Work}   
Though the importance of methods of formal guarantees for neural networks has been well-recognized in literature and there exist several results for verification of feedforward neural networks, especially for ReLU neural networks, only few have been demonstrated to be applicable for dynamical systems that involve neural network components. 

Dutta et al. proposed an MILP based verification engine called Sherlock that performs an output range analysis of ReLU feedforward neural networks in \cite{dutta2017output,dutta2018output}, in which a combined local and global search is developed to more efficiently solve MILP. Also for ReLU  neural networks, the computation of the exact output set of an ReLU neural network is formulated in terms of manipulations of polytopes in \cite{xiang2017reachable_arxiv}. In \cite{xiang2017output}, a concept called maximal sensitivity is introduced to compute the reachable set estimation for feedforward neural networks with a broader class of activation functions satisfying a monotonic assumption. These results are able to produce an estimated/exact output set of a feedforward neural network, and it therefore implies the availability for reachable set estimation and safety verification of neural network control systems. For example, a recent result for piecewise linear systems with ReLU neural network feedback controllers can be found in \cite{xiang2018reachable_acc}, in which the reachability analysis is carried out based on the result of \cite{xiang2017reachable_arxiv}. It has to be mentioned that, as far as we know and as shown in the well-known survey \cite{hunt1992neural}, neural network controllers for complex nonlinear dynamical systems are often designed with activation functions in the forms of tanh, logistic, Gaussian functions other than ReLU that most of the recent neural network verification approaches target. To the best of our knowledge, there is so far no result for reachability and verification of these common classes of neural network control systems, which we present in this paper. Additionally, to the best of our knowledge, there also are no methods for performing reachability analysis and safety verification of closed-loop control systems and CPS with neural network components, which we demonstrate in this paper through alternating computations of neural network reachability and plant reachability.

\subsection{Contributions}
In this paper we develop a novel algorithm to compute the output reachable set of feedforward neural networks with general activation functions. The algorithm is formulated in terms of LP problems.  By incorporating the reachability result for neural networks with reachable set computations for ODE models, we develop a novel algorithm to over-approximate the reachable set of a given neural network control system over a finite time horizon. We present a safety verification algorithm to provide formal safe guarantees for neural network control systems, and demonstrate the methods on an example system using a software prototype we have developed. All of these problems, to the best of our knowledge, have not been previously addressed, and are crucial to solve to enhance assurance of safety in autonomous CPS.

\section{System Description}
This section presents the model of neural network control systems. The neural network controllers considered in this paper are feedforward neural networks with their activation functions in a general form, and plant models are considered in the description of ODEs. Both discrete-time and continuous-time models are presented.
\subsection{Feedforward Neural Networks}

A neural network consists of a number of interconnected neurons. Each neuron is a simple processing element that responds to the weighted inputs it received from other neurons. In this paper, we consider the most popular and general feed-forward neural network, multi-layer perceptrons (MLP). Generally, an MLP consists of three typical classes of layers: An input layer, that serves to pass the input vector to the network, hidden layers of computation neurons, and an output layer composed of at least a computation neuron to generate the output vector.

The action of a neuron depends on its activation function, which is described as 
\begin{align}
u_i = \phi\left(\sum\nolimits_{j=1}^{n}\omega_{ij} \eta_j + \theta_i\right)
\end{align}
where $\eta_j$ is the $j$th input of the $i$th neuron, $\omega_{ij}$ is the weight from the $j$th input to the $i$th neuron, $\theta_i$ is called the bias of the $i$th neuron, $u_i$ is the output of the $i$th neuron, $\phi(\cdot)$ is the activation function. The activation function is generally a nonlinear continuous function  describing the reaction of $i$th neuron with inputs $\eta_j$, $j=1,\cdots,n$. Typical activation functions include ReLU, logistic, tanh, exponential linear unit, linear functions, for instance. In this work, our approach aims at being capable of dealing with activation functions regardless of their specific forms.

An MLP has multiple layers,  each layer $\ell$, $1 \le \ell \le L $, has $n^{[\ell]}$ neurons.  In particular, layer $\ell =0$ is used to denote the input layer and $n^{[0]}$ stands for the number of inputs in the rest of this paper. 
Analogously,  $n^{[L]}$ stands for the last layer, that is the output layer. For a neuron $i$, $1 \le i \le n^{[\ell]}$ in layer $\ell$, the corresponding input vector is denoted by $\boldsymbol{\upeta}^{[\ell]}$ and the weight matrix is 
\begin{equation}
\mathbf{W}^{[\ell]} = \left[\omega_{1}^{[\ell]},\ldots,\omega_{n^{[\ell]}}^{[\ell]}\right]^{\top}
\end{equation}
where $\omega_{i}^{[\ell]}$ is the weight vector. The bias vector for layer $\ell$ is
\begin{equation} 
\boldsymbol {\uptheta}^{[\ell]}=\left[\theta_1^{[\ell]},\ldots,\theta_{n^{[\ell]}}^{[\ell]}\right]^{\top}.
\end{equation} 

The output vector of layer $\ell$ can be expressed as 
\begin{equation}
\mathbf{u}^{[\ell]}=\phi_{\ell}(\mathbf{W}^{[\ell]}\boldsymbol{\upeta}^{[\ell]}+\uptheta^{[\ell]})
\end{equation} 
where $\phi_{\ell}(\cdot)$ is the activation function of layer $\ell$.

For an MLP, the output of $\ell-1$ layer is the input of $\ell$ layer, and the mapping from the input of input layer, that is $\boldsymbol{\upeta}^{[0]}$, to the output of output layer, namely $\mathbf{u}^{[L]}$, stands for the input-output relation of the MLP, denoted by
\begin{equation}\label{NN}
\mathbf{u}^{[L]} = \Phi (\boldsymbol{\upeta}^{[0]})
\end{equation}    
where $\Phi(\cdot) \triangleq \phi_L  \circ \phi_{L - 1}  \circ  \cdots  \circ \phi_1(\cdot) $.

\subsection{Neural Network Control Systems}
In this paper we consider the following discrete-time nonlinear dynamic system in the form of
\begin{align} \label{system}
\left\{ {\begin{array}{*{20}l}
	\mathbf{x}(t+1) = f(\mathbf{x}(t),\mathbf{u}(t))\\
	\mathbf{y}(t) = h(\mathbf{x}(t))
	\end{array} } \right.
\end{align}
where $\mathbf{x}(t) \in \mathbb{R}^{n_x}$ is the system state, $\mathbf{u}(t) \in \mathbb{R}^{n_u}$ and $\mathbf{y}(t) \in \mathbb{R}^{n_y}$ are control input and controlled output, respectively. The general form of nonlinear controllers is 
\begin{align}
\mathbf{u}(t) = \gamma(\mathbf{y}(t),\mathbf{v}(t),t)
\end{align}
where $\mathbf{v}(t) \in \mathbb{R}^{n_v}$ is the reference input for the controller. However, for general $f$ and $h$, the challenging controller design problem is still open even $f$ and $h$ are known. 

Recently, some data-driven approaches have been developed to avoid the difficulties in controller design when system models are complex. One effective method is to use input-output data to train neural networks that are able to generate appropriate control input signals to achieve control objectives. The neural network controller is in the form of 
\begin{align}\label{neural_network_controller}
\mathbf{u}(t) = \Phi(\mathbf{y}(t),\mathbf{v}(t)).
\end{align}

By letting $\boldsymbol{\upeta}(t) = [\mathbf{y}^{\top}(t)~\mathbf{v}^{\top}(t)]^{\top}$, the neural network controller can be rewritten as 
\begin{align}
\label{compact_neural_network_controller}
\mathbf{u}(t) = \Phi(\boldsymbol{\upeta}(t)).
\end{align}

Substituting the above controller into system (\ref{system}), the closed-loop system can be described by
\begin{align}\label{closed_loop}
\left\{ {\begin{array}{*{20}l}
	\mathbf{x}(t+1) = f(\mathbf{x}(t),\Phi(\boldsymbol{\upeta}(t)))\\
	\mathbf{y}(t) = h(\mathbf{x}(t))
	\end{array} } \right.
\end{align}
where $\boldmath{\upeta}(t) = [\mathbf{y}^{\top}(t)~ \mathbf{v}^{\top}(t)]^{\top}$.


Furthermore, we consider continuous-time nonlinear systems in the form of
\begin{align} \label{con_system}
\left\{ {\begin{array}{*{20}l}
	\dot{\mathbf{x}}(t) = f(\mathbf{x}(t),\mathbf{u}(t))\\
	\mathbf{y}(t) = h(\mathbf{x}(t))
	\end{array} } \right. .
\end{align}

In practical applications, there inevitably exists an amount of time for computing the output of neural network (\ref{compact_neural_network_controller}) after receiving input  $\boldsymbol{\upeta} = [\mathbf{y}^{\top}~ \mathbf{v}^{\top}]^{\top}$. Thus, the controller generates controller signals at sampling time instants $t_k$, $k \in \mathbb{N}$, and the value holds between two successive sampling instants $t_k$ and $t_{k+1}$. The sampled neural network controller is given in the following form
\begin{align}\label{con_neural_network_controller}
\mathbf{u}(t) = \Phi(\boldsymbol{\upeta}(t_k)),~t \in [t_k,t_{k+1}).
\end{align}

Substituting the above controller into nonlinear system (\ref{con_system}), the closed-loop system can be obtained as
\begin{align}\label{con_closed_loop}
\left\{ {\begin{array}{*{20}l}
	\dot{\mathbf{x}}(t) = f(\mathbf{x}(t),\Phi(\boldsymbol{\upeta}(t_k)))\\
	\mathbf{y}(t) = h(\mathrm{x}(t))
	\end{array} } \right.,~t \in [t_k,t_{k+1})
\end{align}
where $\boldsymbol{\upeta}(t_k) = [\mathbf{y}^{\top}(t_k)~ \mathbf{v}^{\top}(t_k)]^{\top}$. The system in the form of (\ref{con_closed_loop}) is a sampled-data model for neural network control systems as proposed in \cite{wu2014exponential}. 

The mechanisms of systems (\ref{closed_loop}) and (\ref{con_closed_loop}) are illustrated in Figure \ref{nncs}. As there already exists an extensive literature for reachable set estimation of discrete-time and continuous-time plant models,  the key issue of estimating the reachable set for a neural network control system is to develop methods to estimate the output set of the neural network controller embedded in the closed-loop system. With the solution of estimating output set of neural network controllers, there is no essential difference in reachable set estimation between discrete-time system (\ref{closed_loop}) and continuous-time sample-data system (\ref{con_closed_loop}). The only difference lies in the continuous-time/discrete-time tools that we will employ for computing the reachable set of  plants. Therefore, the rest of this paper will focus on addressing the reachable set estimation problem of discrete-time system (\ref{system}). 

According to the Universal Approximation Theorem \cite{hornik1989multilayer}, it guarantees that, in
principle, such an MLP in the form of (\ref{NN}), namely the function $\Phi(\cdot)$, is able to approximate any nonlinear real-valued function. Despite the impressive ability of approximating nonlinear functions,  much complexities represent in predicting the output behaviors of MLP  (\ref{NN}) as well as systems with neural network controllers (\ref{neural_network_controller}) because of the nonlinearities of MLPs. In the most of real applications, an MLP is usually viewed as a black box to generate  a desirable output with respect to a given input. However, regarding  property verifications such as the safety verification, it has been observed that even a well-trained neural network can react in unexpected and incorrect ways to even slight perturbations of their inputs, which could result in unsafe even unstable systems. Thus, to verify safety properties of dynamical systems with neural network controllers, it is necessary to have a reachable set estimation for the closed-loop system in the form of (\ref{closed_loop}) or (\ref{con_closed_loop}) over a given finite time horizon, which is able to cover all possible values of system state in the given interval, to assure that the state trajectories of the closed-loop system will not attain unreasonable values. 

\begin{figure}
	\begin{center}
	\includegraphics[width=12cm]{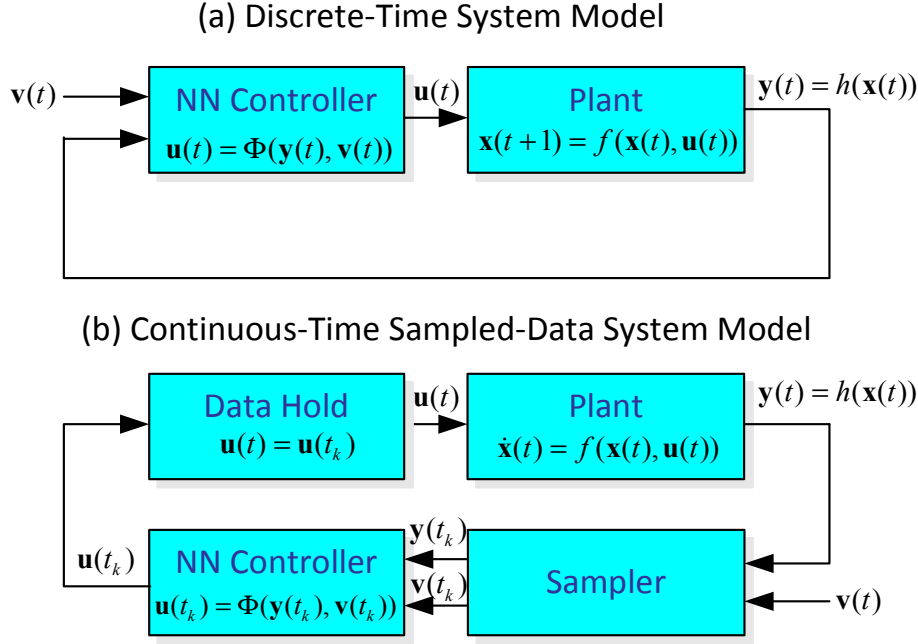}
	\caption{\boldmath Illustration for closed-loop systems with neural network (NN) controllers. Figure (a) is for discrete-time system model in the form of (\ref{closed_loop}) and Figure (b) is for continuous-time systems (\ref{con_closed_loop}) with samplings.  }
	\end{center}
	\label{nncs} 
\end{figure}

\section{Problem Formulation}
We start by defining the MLP output set that will become of interest all through the rest of this paper. With the inputs belonging to a set $\mathcal{H}$, the output  set of an MLP is defined as follows.
\begin{definition}
	Given an MLP in the form of (\ref{NN}) and an input
	set $\mathcal{H}$, the following set
	\begin{align}
	\mathcal{U} = \left\{\mathbf{u} ^{[L]} \in \mathbb{R}^{n_u} \mid \mathbf{u}^{[L]} = \Phi (\boldsymbol{\upeta}^{[0]}),~ \boldsymbol{\upeta}^{[0]} \in \mathcal{H}\right\}
	\end{align}
	is called the output set of MLP (\ref{NN}). 
\end{definition}

Since MLPs are often large, nonlinear, and non-convex, it
is extremely difficult to compute the exact output reachable
set $\mathcal{U}$ for an MLP. Rather than directly computing the exact output reachable set for an MLP, a more practical and feasible way is to derive an over-approximation of $\mathcal{U}$. 

\begin{definition}
	A set $\mathcal{U}_e$ is called an output reachable set estimation of MLP (\ref{NN}), if $\mathcal{U}\subseteq
	\mathcal{U}_e$ holds, where $\mathcal{U}$ is the output
	reachable set of MLP (\ref{NN}).
\end{definition}

The first problem of interest is the reachable set estimation of MLPs in the form of (\ref{NN}).
\begin{problem}\label{problem1}
	Given a bounded input set $\mathcal{H}$ and an MLP
	described by (\ref{NN}), how does one find a set $\mathcal{U}_e$ such that $\mathcal{U} \subseteq \mathcal{U}_e$, and
	make the estimation set $\mathcal{U}_e$ as small as possible\footnote{For a set $\mathcal{Y}$, its over-approximation $\tilde{\mathcal{Y}}_1$ is smaller than another over-approximation $\tilde{\mathcal{Y}}_2$ if 
		$
		d_H(\tilde{\mathcal{Y}}_1,\mathcal{Y}) < d_H(\tilde{\mathcal{Y}}_2,\mathcal{Y})
		$ holds,
		where $d_H$ stands for the Hausdorff distance.}?
\end{problem} 

Let us denote the solution of neural network control system (\ref{closed_loop}) by $\mathbf{x}(t; \mathbf{x}_0,\mathbf{v}(\cdot))$, where $t \in \mathbb{R}_{\ge 0}$ is the time, $\mathbf{x}_0 \in \mathbb{R}^{n_x}$ is the  initial state, $\mathbf{v}(t) \in \mathbb{R}^{n_v}$ is the system input at $t$, $\mathbf{v}(\cdot)$ is the input trajectory. The reachable set of system state at time $t$ can be defined for a set of initial states $\mathcal{X}_0$ and a set of input values $\mathcal{V}$. 
\begin{definition}
	Given a neural network control system in the form of (\ref{closed_loop})  with initial set $\mathcal{X}_0$ and input set $\mathcal{V}$, the reachable set at  time $t$ is
	\begin{equation}
	\mathcal{R}(t) =\left\{\mathbf{x}(t;\mathbf{x}_0,\mathbf{v}(\cdot))\in\mathbb{R}^{n_x} \mid \mathbf{x}_0 \in \mathcal{X}_0, \mathbf{v}(t) \in \mathcal{V}\right\}
	\end{equation}
	and the reachable set over time interval $[t_0,t_f]$ is defined by
	\begin{equation}
	\mathcal{R}([t_0,t_f]) = \bigcup\nolimits_{t \in [t_0,t_f]}\mathcal{R}(t) .
	\end{equation}
\end{definition}

However, exact reachable
sets cannot be computed for most system classes. Especially, it is extremely difficult to compute the exact reachable set $\mathcal{R}(t)$ and $\mathcal{R}([t_0,t_f])$ for systems consisting of neural network components. Similar as Problem \ref{problem1}, instead of computing the exact reachable set $\mathcal{R}(t)$ and $\mathcal{R}([t_0,t_f])$, a more practical and feasible way is to derive over-approximations
of $\mathcal{R}(t)$ and $\mathcal{R}([t_0,t_f])$. The over-approximation results will be sufficient for formal verification.
\begin{definition}\label{def2}
	A set $\mathcal{R}_e(t)$ is an over-approximation of $\mathcal{R}(t)$
	at time $t$ if $\mathcal{R}(t) \subseteq \mathcal{R}_e(t)$ holds.  Moreover, $\mathcal{R}_e([t_0,t_f]) = \bigcup\nolimits_{t \in[t_0,t_f]} \mathcal{R}_e(t)$ is an
	over-approximation of $\mathcal{R}([t_0,t_f])$ over time interval $[t_0, t_f]$.
\end{definition}

Based on Definition \ref{def2}, the problem of reachable set estimation for neural network control system (\ref{closed_loop}) is given as below.
\begin{problem}\label{problem2}
	Given a time $t$, how does one find the set $\mathcal{R}_e(t)$ such that $\mathcal{R}(t) \subseteq \mathcal{R}_e(t)$, given a bounded initial set $\mathcal{X}_0$ and an input set $\mathcal{V}$?
\end{problem}

In this work, we will focus on the safety verification for neural network control systems. The safety specification is expressed by a set defined in the state space, describing the safety requirement.
\begin{definition}
	Safety specification $\mathcal{S}$ formalizes the safety requirements for state $\mathbf{x}(t)$ of neural network control system (\ref{closed_loop}), and is a predicate over state $\mathbf{x}(t)$ of neural network control system (\ref{closed_loop}). The neural network control system (\ref{closed_loop}) is safe over time interval $[t_0, t_f ]$ if and only if the	following condition is satisfied:
	\begin{equation}\label{verification}
	\mathcal{R}([t_0,t_f])  \cap \neg \mathcal{S} = \emptyset
	\end{equation}
	where $\neg$ is the symbol for logical negation.
\end{definition}

Therefore, the safety verification problem for neural network control system (\ref{closed_loop}) is stated as  follows.
\begin{problem}\label{problem3}
	How does one verify the safety requirement described by (\ref{verification}), given a neural network control system (\ref{closed_loop}) with a bounded initial set $\mathcal{X}_0$ and an input set $\mathcal{V}$ and a
	safety specification $\mathcal{S}$?
\end{problem}

Before ending this section, a lemma is presented to show that the safety verification of neural network control system (\ref{closed_loop}) can be relaxed by checking with the over-approximation of reachable set.
\begin{lemma}\label{lemma1}
	Consider a neural network control system in the form of (\ref{closed_loop}) and a safety specification $\mathcal{S}$, the neural network control system is safe in time interval $[t_0,t_f]$ if the following condition is satisfied
	\begin{equation}\label{lemma1_1}
	\mathcal{R}_e([t_0,t_f]) \cap \neg \mathcal{S} = \emptyset
	\end{equation}
	where $\mathcal{R}([t_0,t_f])  \subseteq\mathcal{R}_e([t_0,t_f])$.
\end{lemma}
\begin{proof}
	Due to $\mathcal{R}([t_0,t_f])  \subseteq\mathcal{R}_e([t_0,t_f])$, condition (\ref{lemma1_1}) implies $\mathcal{R}([t_0,t_f]) \cap\neg \mathcal{S} = \emptyset$.
	The proof is complete.
\end{proof}

Lemma \ref{lemma1} implies that the over-approximated reachable set $\mathcal{R}_e([t_0,t_f])$ is qualified for
the safety verification over interval $[t_0,t_f]$. The  three linked problems are the main concerns
to be addressed in the rest of the paper. Essentially, the very first and
basic problem is the Problem \ref{problem1}, namely finding efficient methods to estimate the output set of an MLP. In the remainder of this paper, attentions
are mainly devoted to give solutions for Problem \ref{problem1}, and then extend to solve Problem \ref{problem2}. Problem \ref{problem3} will be considered on the basis of solutions of Problems \ref{problem1} and \ref{problem2}.

\section{Reachability Analysis for Neural Networks}
In this section, we aim at finding ways to over-approximate the output set of an MLP. The input and output sets are formalized as a union of so-called hyper-rectangular sets. An LP based algorithm is developed to efficiently compute an over-approximation of output set of a given MLP.
\subsection{Hyper-Rectangular Sets}
Due to the complex structure and nonlinearities in activation functions, estimating the output reachable set of an MLP represents much difficulties if only using analytical methods. One possible way to circumvent those difficulties is to reduce the input set into a finite number of regularized subregions that are convenient to handle. For a bounded set $\mathcal{H} \in \mathbb{R}^{n}$, we use a finite number of hyper-rectangles to over-approximate it. The hyper-rectangles are constructed as follows: 

For any bounded set $\mathcal{H} \in \mathbb{R}^{n}$, we have $\mathcal{H} \subseteq \mathcal{H}_e$,where $\mathcal{H}_e = \{\boldsymbol{\upeta} \in \mathbb{R}^{n} \mid \underline{\boldsymbol{\upeta}} \le \boldsymbol{\upeta} \le \overline{\boldsymbol{\upeta}}\}$, in which $\underline{\boldsymbol{\upeta}}$ and $\overline{\boldsymbol{\upeta}}$ are defined as the lower and upper bounds of elements of $\boldsymbol{\upeta}$ in $\mathcal{H}$. They are given in detail as
\begin{align}
\underline{\boldsymbol{\upeta}}&=[\mathrm{inf}_{\eta \in \mathcal{H}}(\eta_1),\ldots,\mathrm{inf}_{\eta\in \mathcal{H}}(\eta_n)]^{\top} ,
\\
\overline{\boldsymbol{\upeta}}&=[\mathrm{sup}_{\eta\in \mathcal{H}}(\eta_1),\ldots,\mathrm{sup}_{\eta\in \mathcal{H}}(\eta_n)]^{\top} .
\end{align}

Then, we are able to partition interval \begin{align}
\mathcal{I}_i=[\mathrm{inf}_{\eta\in \mathcal{H}}(\eta_i),~\mathrm{sup}_{\eta\in \mathcal{H}}(\eta_i)],~i \in \{1,\ldots,n\}
\end{align}
into $M_i$ segments as $\mathcal{I}_{i,1}=[\eta_{i,0},\eta_{i,1}]$, $\mathcal{I}_{i,2}=[\eta_{i,1},\eta_{i,2}]$, $\ldots$, $\mathcal{I}_{i,M_i}=[\eta_{i,M_i-1},\eta_{i,M_i}]$, where $\eta_{i,m}$, $m \in \{0,1,\ldots,M_i\}$ are defined 
\begin{align}
\eta_{i,n} = \eta_{i,0} + \frac{m(\eta_{i,M_i}-\eta_{i,0})}{M_i},~m \in \{0,1,\ldots,M_i\}
\end{align} 
with $\eta_{i,0}=\mathrm{inf}_{\eta\in \mathcal{H}}(\eta_i)$
and $\eta_{i,M_i}=\mathrm{sup}_{\eta\in \mathcal{H}}(\eta_i)$. Then,
the hyper-rectangles can be constructed as \begin{align}
\mathcal{P}_i = \mathcal{I}_{1,m_1} \times\cdots\times \mathcal{I}_{n,m_n},~i \in \{1,2,\ldots,\prod\nolimits_{s=1}^{n}M_s\}
\end{align} 
where $\{m_1,\ldots,m_n\}\in \{1,\ldots,M_1\} \times \cdots \times \{1,\ldots,M_n\}$.  

To remove redundant hyper-rectangles, we have to check the existence of intersections between these hyper-rectangles and set $\mathcal{H}$. Hyper-rectangle $\mathcal{P}_i$ satisfying $\mathcal{P}_i \cap \mathcal{H} = \emptyset$ should be removed.
Thus, for any bounded set $\mathcal{H} \in \mathbb{R}^{n}$ with prescribed $M_i$, $i \in \{1,\ldots,n\} $, a finite number of hyper-rectangles $\mathcal{P}_i$, $i \in\{1,\ldots,N\}$ where $1 \le N \le \prod\nolimits_{s=1}^{n}M_s$, can be constructed to over-approximate it.  It can be noted that larger $M_i$, $i \in \{1,\ldots,n\} $ will lead to a preciser over-approximation for the original set $\mathcal{H}$, which will yield better results in the following reachable set estimation, but more hyper-rectangles will be produced which requires more computational costs.  
The hyper-rectangle construction process is summarized by \texttt{hyrec} function in Algorithm \ref{alg1}.

\begin{algorithm}
	\caption{Partition an input set} \label{alg1}
	\begin{algorithmic}[1]
		\Require Set $\mathcal{H}$, partition numbers $M_i$, $i \in \{ 1,\ldots,n\}$
		\Ensure Partition $\mathscr{P} = \{\mathcal{P}_1,\mathcal{P}_2,\ldots,\mathcal{P}_N\}$
		
		\Function{\texttt{hyrec}}{$\mathcal{H}$, $M_i$, $i \in \{ 1,\ldots,n\}$}
		
		\State $\eta_{i,0} \gets \mathrm{inf}_{\eta\in \mathcal{H}}(\eta_i)$, $\eta_{i,M_i} \gets \mathrm{sup}_{\eta\in \mathcal{H}}(\eta_i)$
		
		\For{$i = 1:1:n$}
		\For{$j = 1:1:M_i$}
		\State $\eta_{i,j} \gets \eta_{i,0} + \frac{j(\eta_{i,M_i}-\eta_{i,0})}{M_i}$
		\State $\mathcal{I}_{i,j} \gets [\eta_{i,j-1},\eta_{i,j}]$
		\EndFor
		\EndFor
		
		\State $\mathcal{P}_i \gets \mathcal{I}_{1,m_1} \times\cdots\times \mathcal{I}_{n,m_n}$
		\If{$\mathcal{P}_i \cap \mathcal{H} = \emptyset$}
		\State Remove $\mathcal{P}_i$
		\EndIf
		
		\State\Return $\mathscr{P} = \{\mathcal{P}_1,\mathcal{P}_2,\ldots,\mathcal{P}_N\}$
		
		\EndFunction
	\end{algorithmic}
\end{algorithm}

\subsection{Reachable Set Estimation}
By partitioning the input set $\mathcal{H}$ into $N$ hyper-rectangular sets, the key step for estimating the  output set of an MLP is to find a way to estimate the output set for each individual hyper-rectangle as the input to the MLP. First, we consider a single layer.
\begin{theorem}\label{thm1}
	Consider a single layer $
	\mathbf{u}=\phi(\mathbf{W}\boldsymbol{\upeta}+\boldsymbol{\uptheta})
	$, if the input set is a hyper-rectangle described by  $\mathcal{I}_{1} \times\cdots\times \mathcal{I}_{n_\eta}$ where $\mathcal{I}_{i} = [\underline{\eta}_i,\overline \eta_i]$, $i \in \{1,\ldots,n_\eta\}$, the output set can be over-approximated by a hyper-rectangle in the expression of intervals $\mathcal{I}_{1} \times\cdots\times \mathcal{I}_{n_u}$, where $\mathcal{I}_{i}$, $i \in \{1,\ldots,n_u\}$, can be computed by 
	$
	\mathcal{I}_{i} = 	[\underline{u}_i ,\overline u_i]
	$,
	where $\underline{u}_i$ and  $\overline{u}_i$ are computed by
	\begin{align}
	\underline{u}_i &= \min_{\boldsymbol{\upeta} \in \mathcal{H}} \phi\left(\sum\nolimits_{j=1}^{n_\eta}\omega_{ij} \eta_j + \theta_i\right), \label{thm1_1}
	\\
	\overline u_i &= \max_{\boldsymbol{\upeta} \in \mathcal{H}} \phi\left(\sum\nolimits_{j=1}^{n_\eta}\omega_{ij} \eta_j + \theta_i\right) . \label{thm1_2}
	\end{align} 
\end{theorem}
\begin{proof}
	Since activation function $\phi(\cdot)$ is continuous and hyper-rectangle $\mathcal{I}_{1} \times\cdots\times \mathcal{I}_{n_\eta}$ is bounded, there exist $\underline{u}_i$ and $\overline u_i$ defined by (\ref{thm1_1}) and (\ref{thm1_2})  such that $\underline{u}_i \le u_i \le \overline u_i$ holds for any $i \in \{1,\ldots,n_u\}$. The proof is complete. 
\end{proof}
To obtain the output hyper-rectangle, we need to compute the minimum and maximum values of the output of nonlinear function $\phi$. For general nonlinear functions, the optimization problems are still challenging. Typical activation functions include ReLU, logistic, tanh, exponential linear unit, linear functions, for instance, satisfy the following monotonic assumption.

\begin{assumption}\label{assumption_1}
	For any two scalars $z_1 \le z_2$, the activation function satisfies $\phi(z_1) \le \phi(z_2)$. 
\end{assumption}

Assumption \ref{assumption_1} is a common property that can be satisfied by a variety of activation functions. For example, it is easy to verify that the most commonly used such as logistic, tanh,  ReLU, all satisfy Assumption \ref{assumption_1}.  Taking advantage of the monotonic property, optimization problems (\ref{thm1_1}) and (\ref{thm1_2}) can be turned into LP problems which can be solved efficiently with the aid of existing tools.

\begin{theorem}\label{thm2}
	Consider a single layer $
	\mathbf{u}=\phi(\mathbf{W}\boldsymbol{\upeta}+\boldsymbol{\uptheta})
	$
	with activation functions satisfying Assumption \ref{assumption_1}, if the input set is a hyper-rectangle described by  $\mathcal{I}_{1} \times\cdots\times \mathcal{I}_{n_\eta}$ where $\mathcal{I}_{i} = [\underline{\eta}_i,\overline\eta_i]$, $i \in \{1,\ldots,n_\eta\}$, the output set can be over-approximated by a hyper-rectangle in the expression of intervals $\mathcal{I}_{1} \times\cdots\times \mathcal{I}_{n_u}$, where $\mathcal{I}_{i}$, $i \in \{1,\ldots,n_u\}$, can be computed by 
	$
	\mathcal{I}_{i} = 	[\phi(\underline{z}_i),\phi(\overline{z}_i)]
	$,	where $\underline{z}_i$ and  $\overline z_i$ are the solutions of the LP problems
	\begin{align}
	&\min \underline{z}_i \nonumber
	\\ \mathrm{s.t.}~~&  
	\underline{z}_i = \sum\nolimits_{j=1}^{n_\eta}\omega_{ij} \eta_j + \theta_i, \nonumber
	\\
	&\underline{\eta}_j \le \eta_j \le \overline{\eta}_j,~j \in \{1,\ldots,n_\eta\}, \label{thm2_1}
	\end{align}
	\begin{align} 
	&\max \overline z_i  \nonumber
	\\
	\mathrm{s.t.}~~&\overline z_i = \sum\nolimits_{j=1}^{n_\eta}\omega_{ij} \eta_j + \theta_i \nonumber
	\\
	&\underline{\eta}_j \le \eta_j \le \overline{\eta}_j,~j \in \{1,\ldots,n_\eta\}  . \label{thm2_2}
	\end{align} 
\end{theorem} 
\begin{proof}
	From (\ref{thm2_1}) and (\ref{thm2_2}), one can obtain
	\begin{align}
	\underline{z}_i	\le \sum\nolimits_{j=1}^{n_\eta}\omega_{ij} \eta_j + \theta_i \le \overline z_i
	\end{align}
	where $\underline{z}_i$ and $\overline z_i$ are the solutions of (\ref{thm2_1}) and (\ref{thm2_2}), respectively. Then, using the monotonic Assumption \ref{assumption_1} to get 
	\begin{align}
	\phi(\underline{z}_i)\le \phi(\sum\nolimits_{j=1}^{n_\eta}\omega_{ij} \eta_j + \theta_i) \le \phi(\overline z_i) .
	\end{align}
	Thus, the proof is complete. 
\end{proof}

Theorem \ref{thm2} uses the monotonic property of activations to avoid  complex nonlinear optimization problems. Actually, this idea can be also applied to other types of activation functions even when the monotonic condition does not hold. For example, if we consider Gaussian function $\phi(z) = e^{-z^2}$, the following result can be obtained by the similar guidelines in Theorem \ref{thm2}. 

\begin{theorem}\label{thm3}
	Consider a single layer $
	\boldsymbol{u} = \phi(\mathbf{W}\boldsymbol{\upeta}+\boldsymbol{\uptheta})
	$
	with $\phi(z) = e^{-z^2}$, if the input set is a hyper-rectangle described by  $\mathcal{I}_{1} \times\cdots\times \mathcal{I}_{n_\eta}$ where $\mathcal{I}_{i} = [\underline{\eta}_i,\overline \eta_i]$, $i \in \{1,\ldots,n_\eta\}$, the output set can be over-approximated by a hyper-rectangle in the expression of intervals $\mathcal{I}_{1} \times\cdots\times \mathcal{I}_{n_u}$, where $\mathcal{I}_{i}$, $i \in \{1,\ldots,n_u\}$, can be computed by 
	\begin{align}\label{cor1_1}
	\mathcal{I}_{i} = \left\{
	{\begin{array}{*{20}l}
		{[\phi(\underline{z}_i),\phi(\overline z_i)],} & {\overline z_i  \le 0}  \\
		{[\phi(\underline{z}_i),1],} & {\underline{z}_i <0,~\overline z_i >0,~\underline{z}_i+\overline z_i \le 0}  \\
		{[\phi(\overline z_i),1],} & {\underline{z}_i <0,~\overline z_i >0,~\underline{z}_i+\overline z_i > 0}  \\
		{[\phi(\overline z_i),\phi(\underline{z}_i)],} & {\underline{z}_i \ge 0}  \\
		\end{array} } \right.
	\end{align}
	where $\underline{z}_i$ and  $\overline z_i$ are the solutions of LP problems (\ref{thm2_1}) and (\ref{thm2_2}), respectively.
\end{theorem}
\begin{proof}
	In the interval $(-\infty,0]$,  activation function $\phi(z) = e^{-z^2}$ is monotonically increasing and the hyper-rectangle $\mathcal{I}_i = [\phi(\underline{z}_i),\phi(\overline z_i)]$ can be obtained same as in Theorem \ref{thm1} for $z_i \in [\underline{z}_i,\overline{z}_i] \subset (-\infty,0]$ where $z_i =\sum\nolimits_{j=1}^{n_\eta}\omega_{ij}\eta_j + \theta_i$.
	
	Similarly, for $[0,+\infty)$, the activation function $\phi(z) = e^{-z^2}$ is monotonically decreasing. In this case, the monotonically decreasing property of $\phi$ directly yields $\mathcal{I}_i = [\phi(\overline z_i),\phi(\underline{z}_i)]$ by following the similar arguments of monotonically increasing case.
	
	Moreover, for the case of $z_i \in [\underline{z}_i,\overline{z}_i]$ where $\underline{z}_i <0$ and $\overline z_i >0$, the maximum value of $\phi(z_i)$ can be obtained with $z_i=0$, that is $\phi(0)=1$. Then, due to the monotonic property in intervals $(-\infty,0]$ and $[0,\infty)$, the minimum value of $\phi(z_i)$, $z_i \in [\underline{z}_i,\overline{z}_i]$ should be $\phi(\underline{z}_i)$ if $\left|\underline{z}_i \right|>\left| \overline z_i\right|$, and otherwise be $\phi(\overline{z}_i)$ if $\left|\underline{z}_i\right| \le \left| \overline z_i \right|$. The proof is complete. 
\end{proof}

The LP problems (\ref{thm2_1}) and (\ref{thm2_2}) can be efficiently solved by
existing tools such as \texttt{linprog} in Matlab. Actually, to be more efficient in output set estimation computation without evaluating the functions involved in (\ref{thm2_1}) and (\ref{thm2_2}), the solutions $\underline{z}_i$ and $\overline{z}_i$ can be explicitly written out  as
\begin{align}
\underline{z}_i & = \sum\nolimits_{j=1}^{n_\eta}\underline{g}_{ij}
\\
\overline{z}_i &= \sum\nolimits_{j=1}^{n_\eta}\overline{g}_{ij}
\end{align}
with $\underline{g}_{ij}$ and $\overline{g}_{ij}$ defined by
\begin{align}
\underline{g}_{ij}  &= \left\{ {\begin{array}{*{20}l}
	{\omega _{ij} \underline{\eta}_j +  \theta_i,} & {\omega _{ij}\geq 0}  \\
	{\omega _{ij} \overline \eta_j +  \theta_i,} & {\omega _{ij}  < 0}  \\
	\end{array} } \right.
\\
\overline g_{ij}& = \left\{ {\begin{array}{*{20}c}
	{\omega _{ij} \overline \eta_j +  \theta_i,} & {\omega _{ij}  \geq 0}  \\
	{\omega _{ij} \underline{\eta}_j +  \theta_i ,} & {\omega _{ij}  < 0}  \\
	\end{array} } \right..
\end{align}

Theorem \ref{thm1} indicates that the output set of one single layer can be over-approximated by a hyper-rectangle as long as the input set is a hyper-rectangle. Moreover, Theorems \ref{thm2} and \ref{thm3} turn the computation process into LP problems. For multi-layer neural networks in the form of (\ref{neural_network_controller}), that is $\mathbf{u}^{[L]} = \Phi(\boldsymbol{\upeta}^{[0]}) = \phi_L  \circ \phi_{L - 1}  \circ  \cdots  \circ \phi_1(\boldsymbol{\upeta}^{[0]})$, a layer-by-layer approach can be developed based on these results. For an MLP, it essentially has $\boldsymbol{\upeta}^{[\ell]}=\mathrm{u}^{[\ell-1]}$, $\ell=1,\ldots,L$. Given an input set as a union of hyper-rectangles, Theorem \ref{thm1} assures the input set of every layer, which is the output set of the preceding layer, can be over-approximated by a set of hyper-rectangles. Therefore, the computation can proceed in a layer-by-layer manner and the output set of layer $L$ becomes an over-approximation of the output set of the MLP. 
Taking Theorem \ref{thm2} for instance, function \texttt{reachMLP} given in Algorithm \ref{alg2} illustrates the layer-by-layer approach for over-approximating the output set of an MLP under Assumption \ref{assumption_1}. A similar algorithm can be developed for Theorem \ref{thm3}, which is omitted here.  

\begin{algorithm} [t!]
	\caption{Output set over-approximation for MLP with monotonic increasing activation functions} \label{alg2}
	\begin{algorithmic}[1]
		\Require Weight matrices $\mathbf{W}^{[\ell]}$, bias $\boldsymbol{\uptheta}^{[\ell]}$, $\ell \in \{ 1,\ldots,L\}$, set $\mathcal{H}$, partition numbers $M_i$, $i \in \{ 1,\ldots,n\}$
		\Ensure Output set over-approximation $\mathcal{U}_e$
		
		\Function{\texttt{reachMLP}}{$\mathbf{W}^{[\ell]}$, $\boldsymbol{\uptheta}^{[\ell]}$, $\ell \in \{ 1,\ldots,L\}$, $\mathcal{H}$, $M_i$, $i \in \{ 1,\ldots,n\}$}
		
		\State $\mathscr{P} \gets \texttt{hyrec}(\mathcal{H}, M_i, i \in \{ 1,\ldots,n\})$
		
		\For{$p=1:1:\left|\mathscr{P}\right|$}
		\State $\mathcal{I}_1^{[1]}\times\cdots\times\mathcal{I}_{n^{[1]}}^{[1]} \gets \mathcal{P}_p$
		\For{$j=1:1:L$}
		\For{$i=1:1:n^{[j]}$}
		\State $\underline{g}_{ij}  \gets \left\{ {\begin{array}{*{20}l}
			{\omega _{ij} \underline{\eta}_j +  \theta_i,} & {\omega _{ij}\geq 0}  \\
			{\omega _{ij} \overline \eta_j +  \theta_i,} & {\omega _{ij}  < 0}  \\
			\end{array} } \right.$
		\State $\overline g_{ij} \gets \left\{ {\begin{array}{*{20}c}
			{\omega _{ij} \overline \eta_j +  \theta_i,} & {\omega _{ij}  \geq 0}  \\
			{\omega _{ij} \underline{\eta}_j +  \theta_i ,} & {\omega _{ij}  < 0}  \\
			\end{array} } \right.$
		\State $\underline{z_i} \gets \sum\nolimits_{j=1}^{n_\eta}\underline{g}_{ij}$ 
		\State $\overline{z_i} \gets \sum\nolimits_{j=1}^{n_\eta}\overline{g}_{ij}$
		\State 	$\mathcal{I}_{i}^{[j+1]} \gets 	[\phi (\underline{z}_i ) ,~\phi (\overline{z}_i ) ]$		
		\EndFor		
		\EndFor
		\State $\mathcal{U}_{e,p} \gets \mathcal{I}_1^{[L]}\times\cdots\times\mathcal{I}_{n^{[L]}}^{[L]}$
		\EndFor
		\State $\mathcal{U}_e \gets \bigcup\nolimits_{p=1}^{\left|\mathscr{P}\right|}\mathcal{U}_{e,p}$
		\State \Return $\mathcal{U}_e$
		\EndFunction
	\end{algorithmic}
\end{algorithm}


\subsection{Example}\label{example1}
Let us consider an MLP with two inputs, two outputs and one hidden layer consisting of 7 neurons. The activation function for the hidden layer is  \texttt{tanh} function, and \texttt{purelin} function is for the output layer. The weight matrices and bias vectors are randomly generated as listed below:

\begin{align*}
&\mathbf{W}^{[1]}=\left[ {\begin{array}{*{20}c}
	-1.0927  & -0.9738 \\
	0.0974 &  -1.6347  \\
	-1.3900  &  1.3535  \\
	0.2311  &  3.2967  \\
	0.1067 &  -0.4837  \\
	-0.1264 &   0.3281  \\
	0.8038  &  0.5583  \\
	\end{array} } \right],~\boldsymbol{\uptheta}^{[1]}=\left[ {\begin{array}{*{20}c}
	0.7752 \\
	-0.7823 \\
	-0.5119 \\
	0.3074 \\
	-0.7417 \\
	0.7618 \\ 
	1.2038 \\
	\end{array} } \right],
\\
&\mathbf{W}^{[2]}=\left[ {\begin{array}{*{20}c}
	1.5441  &   -1.0941 \\
	1.4009  & -0.7114 \\
	-0.9595 &   0.5236 \\
	-0.4089 &  -0.7377 \\
	0.3599  & -0.7392 \\
	0.0068  &  0.1388 \\ 
	-0.2026 &   0.0655  \\
	\end{array} } \right]^{\top},~
\boldsymbol{\uptheta}^{[2]}=\left[ {\begin{array}{*{20}c}
	0.2315 \\
	-0.3555 \\
	\end{array} } \right].
\end{align*}

The input set is considered in a unit box, that is 
\begin{equation*}
\mathcal{H} = \{\boldsymbol{\upeta} \in \mathbb{R}^{2} \mid \left\|\boldsymbol{\upeta}\right\|_{\infty} \le 1\}.
\end{equation*}

First, the partition number is selected to be $M_i = 10$, which implies there are total 100 rectangular reachtubes generated for output set estimation. The estimated reachable set is illustrated in Figure \ref{fig1}. From Figure \ref{fig1}, it can be seen that the estimated output set of the proposed MLP is a union of two dimensional rectangles which is able to include all output values generated from the input set $\mathcal{H}$.

Furthermore, to show how the partition number affects the estimation result and obtain a preciser estimation for the output set, different partitioning numbers discretizing set $\mathcal{H}$ are used to apply Algorithm \ref{alg2}, they are $M_i = 10,20,30,40,50$. The estimated reachable sets are shown in Figure \ref{fig1_2}. With  finer discretizations, more computation efforts are required for running function \texttt{reachMLP}. The computation time and  number of reachtubes  with different partition numbers are listed in Table \ref{tab1}.  Comparing those results, it can be observed that larger partition numbers can lead to better estimation results at the expense of more computation efforts. 
\begin{figure}
	\begin{center}
	\includegraphics[width=12cm]{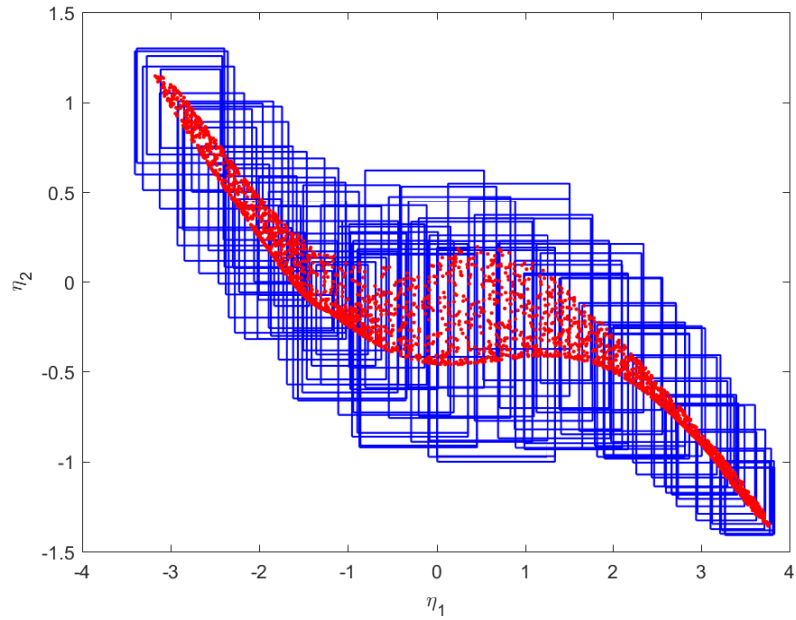}
	\caption{\boldmath Output  set estimation with $M_1=M_2=10$. Input being the unit box with partition number $M_i =10$ leads to 100 rectangular reachtubes (blue rectangles) constituting the over-approximated output set. $5000$ random outputs (red spots) are all located in the estimated output set, namely the union of the 100 reachtubes.}
	\end{center}
	\label{fig1}
\end{figure}
\begin{figure}
	\begin{center}
	\includegraphics[width=12cm]{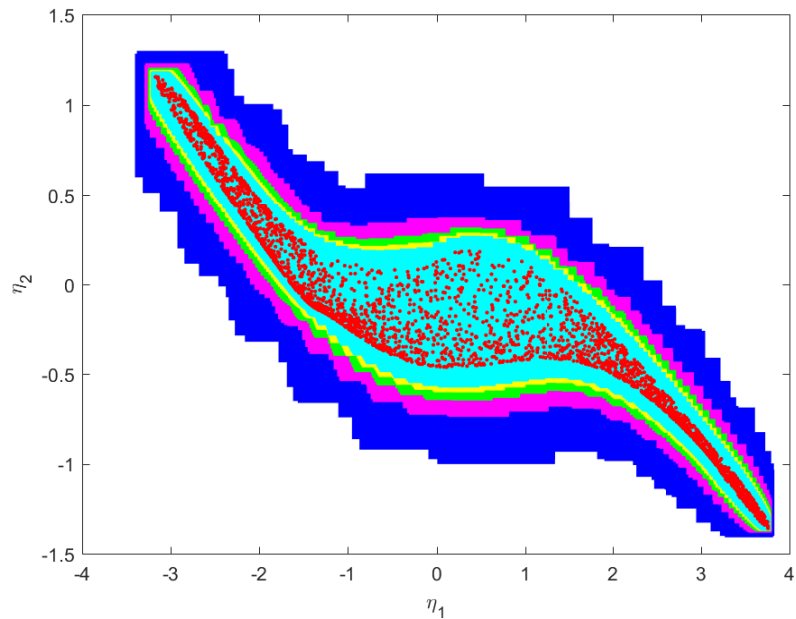}
	\caption{\boldmath Output  set estimation results with different partition number: $M_1 =M_2 =10$ (union of cyan,  yellow, green, magenta and blue areas), $M_1 =M_2 =20$ (union of cyan, yellow, green and  magenta areas), $M_1 =M_2 =30$ (union of cyan, yellow and green areas), $M_1 =M_2 =40$ (union of cyan and yellow areas) and $M_1 =M_2 =50$ (cyan area). Tighter over-approximations can be obtained with larger partition numbers, and $5000$ random outputs (red spots) are all located in all the estimated reachable sets. }
\end{center}
	\label{fig1_2}
\end{figure}
\begin{table}[t!]
	\centering
	\caption{Computation time and number of reachtubes  with different partition numbers}\label{tab1}
	\begin{tabular}{ccc}
		\hline
		Partition Number & Reachtube Number &  Computation Time \\ 
	   \hline
		$M_1=M_2=10$   & 100 & 0.062304 seconds\\
		
		$M_1=M_2=20$ & 400 &  0.074726 seconds \\
		
		$M_1=M_2=30$ & 900 &  0.142574 seconds\\
		
		$M_1=M_2=40$  & 1600 &   0.251087 seconds\\
		
		$M_1=M_2=50$  & 2500 &  0.382729 seconds\\
		\hline
	\end{tabular}
\end{table} 

\section{Reachability and Verification for Neural Network Control Systems}
This section presents the reachability analysis and safety verification algorithms for neural network control systems. The developed algorithms combine the aforementioned output set computation result for MLPs and existing reachable set estimation methods for ODE models.  
\subsection{Reachability Analysis and Safety Verification}
Based on Algorithm \ref{alg2}, which is able to obtain an over-approximation for the output set of a given MLP, the reachable set estimation for neural network control systems in the form of (\ref{closed_loop}) can be performed. The procedure generally involves two parts. Algorithm \ref{alg2} is utilized to compute an over-approximation of the output set for a neural network controller. We also have to compute the reachable set and output set of system (\ref{system}) with the input set obtained by Algorithm \ref{alg2}. For reachable set computation of systems described by ordinary difference/differential equations (ODE), there exist various approaches and tools such as those well-developed in \cite{frehse2011spaceex,chen2013flow,duggirala2015c2e2,bak2017hylaa}. In this paper, we shall not develop new methods to estimate reachable sets of ODE models since an extensive literature
is by now available for this problem on both discrete-time and continuous-time models. The following function is given to denote the reachable set estimation for ODE models
\begin{align}
\mathcal{R}_e(t+1) & = \texttt{reachODEx}(f,\mathcal{U}(t),\mathcal{R}_e(t))
\\
\mathcal{Y}_e(t) & = \texttt{reachODEy}(h,\mathcal{R}_e(t))
\end{align}
where $\mathcal{U}(t)$ is the input set, $\mathcal{R}_e(t+1)$ is the estimated reachable set for state $\mathbf{x}(t+1)$ and $\mathcal{Y}_e(t)$ is the estimated reachable set for output $\mathbf{y}(t)$. 

Recursively using  $\texttt{reachMLP}$ proposed by Algorithms \ref{alg2} together with $\texttt{reachODEx}$, $\texttt{reachODEy}$, an over-approximation of the reachable set of a closed-loop system with a neural network controller can be obtained, which is summarized by Proposition \ref{pro1} and Algorithm \ref{alg3}. 
\begin{algorithm}
	\caption{Reachable set estimation for dynamical systems with neural network controllers} \label{alg3}
	\begin{algorithmic}[1]
		\Require  System dynamics $f$ and output equation $h$, weight matrices $\mathbf{W}^{[\ell]}$, bias $\boldsymbol{\uptheta}^{[\ell]}$, $\ell = 1, \ldots,L$, initial set $\mathcal{X}_0$, input set $\mathcal{V}$, partition numbers $M_i, i \in \{ 1,\ldots,n^{[0]}\}$
		\Ensure Reachable set estimation $\mathcal{R}_e([t_0,t_f])$.
		
		\Function{\texttt{reachNNCS}}{$f$, $h$, $\mathbf{W}^{[\ell]}$, $\boldsymbol{\uptheta}^{[\ell]}$, $\ell = 1, \ldots,L$, $\mathcal{X}_0$, $\mathcal{V}$, $M_i, i \in \{ 1,\ldots,n^{[0]}\}$}
		\State $\mathcal{R}_e(t_0) \gets \mathcal{X}_0 $
		\While{$t \le t_f$}
		\State $\mathcal{Y}_e(t)\gets \texttt{reachODEy}(h,\mathcal{R}_e(t))$
		\State $\mathcal{H} \gets \mathcal{Y}_e(t) \times \mathcal{V}$
		\State $\mathcal{U}_e(t) \gets \texttt{reachMLF}(\mathbf{W}^{[\ell]}, \boldsymbol{\uptheta}^{[\ell]}, \mathcal{H},M_i)$	
		\State $\mathcal{R}_e(t+1)\gets \texttt{reachODEx}(f,\mathcal{U}_e(t),\mathcal{R}_e(t))$	
		\EndWhile
		
		\State 	$\mathcal{R}_e([t_0,t_f]) \gets \bigcup\nolimits_{t \in [t_0,t_f]}\mathcal{R}_{e}(t)$
		\State \Return $\mathcal{R}_e([t_0,t_f])$ 
		\EndFunction
	\end{algorithmic}
\end{algorithm}

\begin{proposition}\label{pro1}
	Consider neural network control system in the form of (\ref{closed_loop}) with initial set $\mathcal{X}_0$,  input set $\mathcal{V}$, the reachable set $\mathcal{R}([t_0,t_f])$ satisfies $\mathcal{R}([t_0,t_f])\subseteq \mathcal{R}_e([t_0,t_f])$, where $\mathcal{R}_e([t_0,t_f])$ is an estimated reachable set obtained by Algorithm \ref{alg3}.
\end{proposition}

For continuous-time neural network control systems in the form of  (\ref{con_closed_loop}), a similar procedure can be established to perform the reachable set estimation as Algorithm \ref{alg3}. Likewise, two parts are involved for interval $[t_k,t_{k+1})$. At each beginning sampling instant $t_k$, it is the time for computing the output set of a neural network controller. Algorithm \ref{alg2} is utilized to compute an over-approximation of the output  set for the neural network controller. Then, it is noted that the output produced by the neural network controller maintains its value during the interval $[t_k,t_{k+1})$, thus the reachable set estimation for the proposed nonlinear continuous-time system can be computed by using some sophisticated methods or tools such as \cite{frehse2011spaceex,chen2013flow,duggirala2015c2e2,bak2017hylaa}, given the unchanged set of a neural network controller in $(t_k,t_{k+1})$.  

Based on the estimated reachable set obtained by Algorithm \ref{alg4} and using Lemma \ref{lemma1}, the safety property can be examined with the existence of intersections between estimated reachable set and unsafe region $\neg \mathcal{S}$.

\begin{proposition}\label{pro2}
	Consider a neural network control system in the form of (\ref{closed_loop}) with a safety specification $\mathcal{S}$, the system is safe in $[t_0,t_f]$, if $ \mathcal{R}_e([t_0,t_f]) \cap \neg \mathcal{S} = \emptyset$, where $\mathcal{R}_e([t_0,t_f]) = \texttt{reachNNCS}(\mathbf{W}^{[\ell]}$, $\boldsymbol{\uptheta}^{[\ell]}$, $\ell = 1, \ldots,L$, $\mathcal{X}_0$, $\mathcal{V}$, $M_i, i \in \{ 1,\ldots,n^{[0]}\}$, $\mathcal{S})$ is obtained by Algorithm \ref{alg4}.
\end{proposition}

Function \texttt{verifyNNCS} is developed based on Algorithm \ref{alg3} for Problem \ref{problem3}, the safety verification problem for neural network control systems.  If function \texttt{verifyNNCS} returns SAFE, then the neural network control system is safe. If it returns UNCERTAIN, that means the safety property is unclear for this case.
\begin{algorithm}
	\caption{Safety verification for systems with neural network controllers} \label{alg4}
	\begin{algorithmic}[1]
		\Require System dynamics $f$ and output equation $h$, weight matrices $\mathbf{W}^{[\ell]}$, bias $\boldsymbol{\uptheta}^{[\ell]}$, $\ell = 1, \ldots,L$, initial set $\mathcal{X}_0$, input set $\mathcal{V}$, partition numbers $M_i, i \in \{ 1,\ldots,n^{[0]}\}$, safety specification $\mathcal{S}$
		\Ensure SAFE or UNCERTAIN
		
		\Function{\texttt{verifyNNCS}}{$\mathbf{W}^{[\ell]}$, $\boldsymbol{\uptheta}^{[\ell]}$, $\ell = 1, \ldots,L$, $\mathcal{X}_0$, $\mathcal{V}$, $M_i, i \in \{ 1,\ldots,n^{[0]}\}$, $\mathcal{S}$}
		
		\State  $\mathcal{R}_e([t_0,t_f]) \gets \texttt{reachNNCS}(f,h,\mathbf{W}^{[\ell]},\boldsymbol{\uptheta}^{[\ell]},\mathcal{X}_0,\mathcal{V},M_i)$
		\If{$\mathcal{R}_e([t_0,t_f])  \cap \mathcal{S} = \emptyset$}
		\State \Return SAFE
		\Else 
		\State \Return UNCERTAIN
		\EndIf
		\EndFunction
	\end{algorithmic}
\end{algorithm}
\subsection{Example}
Consider a discrete-time linear system 
\begin{align*} 
\left\{ {\begin{array}{*{20}l}
	\mathbf{x}(t+1) = \mathbf{A}\mathbf{x}(t)+\mathbf{B u }(t)\\
	\mathbf{y}(t) = \mathbf{Cx}(t)
	\end{array} } \right.
\end{align*}
where $\mathbf{x}(t) \in \mathbb{R}^{2}$, $\mathbf{u}(t) \in \mathbb{R}$ and $\mathbf{y}(t) \in \mathbb{R}$. System matrices $\mathbf{A} \in \mathbb{R}^{2 \times 2}$, $\mathbf{B} \in \mathbb{R}^{2 \times 1}$ and $\mathbf{C} \in \mathbb{R}^{1 \times 2}$ are randomly chosen as 
\begin{align*}
\mathbf{A}=\left[ {\begin{array}{*{20}c}
	-0.6722  &  0.0935 \\
	-0.4011  &  0.4969  \\
	\end{array} } \right],~
\mathbf{B}=\left[ {\begin{array}{*{20}c}
	0.4805 \\
	-0.3911  \\
	\end{array} } \right],~
\mathbf{C}=\left[ {\begin{array}{*{20}c}
	-0.4625  \\  1.4874 \\
	\end{array} } \right].
\end{align*}

The neural network feedback law is $\mathbf{u }(t) = \Phi(\mathbf{y}(t),\mathbf{v}(t))$ with $\mathbf{v} \in \mathcal{V} = \{\mathbf{v} \in \mathbb{R}\mid \left\|\mathbf{v}\right\|_\infty \le 0.5\}$. The neural network is chosen with $1$ hidden layer consisting of $5$ neurons, and the weight matrices and bias vectors are also randomly generated as
\begin{align*}
&\mathbf{W}^{[1]}=\left[ {\begin{array}{*{20}c}
	0.8530  & -1.0127 \\
	-1.1751 &  -1.2403  \\
	-0.4544 &   0.2666 \\
	-0.5061 &  -1.2078 \\
	1.8037  & -1.0501 \\
	\end{array} } \right],~\boldsymbol{\uptheta}^{[1]}=\left[ {\begin{array}{*{20}c}
	0.7996 \\
	1.6286 \\
	0.1291 \\
	-2.0848 \\
	-0.6471\\
	\end{array} } \right],
\\
&\mathbf{W}^{[2]}=\left[ {\begin{array}{*{20}c}
	0.4687  \\  0.2829 \\   1.3412 \\   0.3806  \\  1.4354
	\end{array} } \right]^{\top},~
\boldsymbol{\uptheta}^{[2]}=
-0.2517 .
\end{align*}

Assume the initial state set as 
\begin{align*}
\mathcal{X}_0 = \{\mathbf{x} \in \mathbb{R}^2\mid \left\|\mathbf{x}-\mathbf{x}_{c,1}\right\|_\infty \le 0.5\}
\end{align*} 
where $x_{c,1} = [2.5~~2.5]^{\top}$. To estimate the reachable set of system state $\mathbf{x}(t)$, $t \in [0,10]$, we execute Algorithm \ref{alg3} in which function $\texttt{reachMLF}$ is described by Algorithm \ref{alg2}. For the reachable set $\mathcal{R}_e(t)$ at each step $t$, we use the convex hull of all reachable sets produced by the reachtubes of the neural network controller, denoted by $\tilde{\mathcal{R}}_e(t)$ such that $\mathcal{R}_e(t) \subseteq \tilde{\mathcal{R}}_e(t)$. As a result, functions $\texttt{reachODEx}$ and $\texttt{reachODEy}$ can  be expressed as
\begin{align*}
\texttt{reachODEx}(f,\mathcal{U}(t),\tilde{\mathcal{R}}_e(t)) &= \mathbf{A} \tilde{\mathcal{R}}_e(t) + \mathbf{B}\mathcal{U}(t),\\
\texttt{reachODEy}(h,\tilde{\mathcal{R}}_e(t)) &= \mathbf{C} \tilde{\mathcal{R}}_e(t)
\end{align*} 
which can be efficiently computed in the employment of sophisticated computational convex geometry tools such as MPT3 \cite{MPT3}. Since  $\mathcal{R}_e(t) \subseteq \tilde{\mathcal{R}}_e(t)$, one can  obtain
\begin{align*}
\mathcal{R}_e(t+1)&\subseteq\texttt{reachODEx}(f,\mathcal{U}(t),\tilde{\mathcal{R}}_e(t)),
\\
\mathcal{Y}_e(t)&\subseteq\texttt{reachODEy}(h,\tilde{\mathcal{R}}_e(t)).
\end{align*} 

The estimation results are shown in Figures \ref{fig3} and \ref{fig4}. Similar as in Example \ref{example1}, we use different partition numbers to observe the differences in the estimation results. Two partition numbers $M_i = 5,20$ are used, and obviously a larger partition number leads to a tighter estimation result which is consistent with the observation in Example \ref{example1}. 

Furthermore, with the estimated reachable set, the safety property can be easily verified by inspecting  the figures for non-empty intersections between the over-approximation of the reachable set and unsafe regions.  For instance, assuming the unsafe region as 
\begin{equation*}
\neg \mathcal{S} = \{\mathbf{x}\in \mathbb{R}^2 \mid \left\|\mathbf{x}-\mathbf{x}_{c,2}\right\|_\infty \le 0.5\}
\end{equation*}
where $\mathbf{x}_{c,2} = [-2,5~~2.5]^{\top}$, the safety of the system can be claimed for interval $[0,10]$ since there exists no intersection of the estimated reachable set and the unsafe region as shown in Figure \ref{fig2}.

\begin{figure}
	\begin{center}
		\includegraphics[width=12cm]{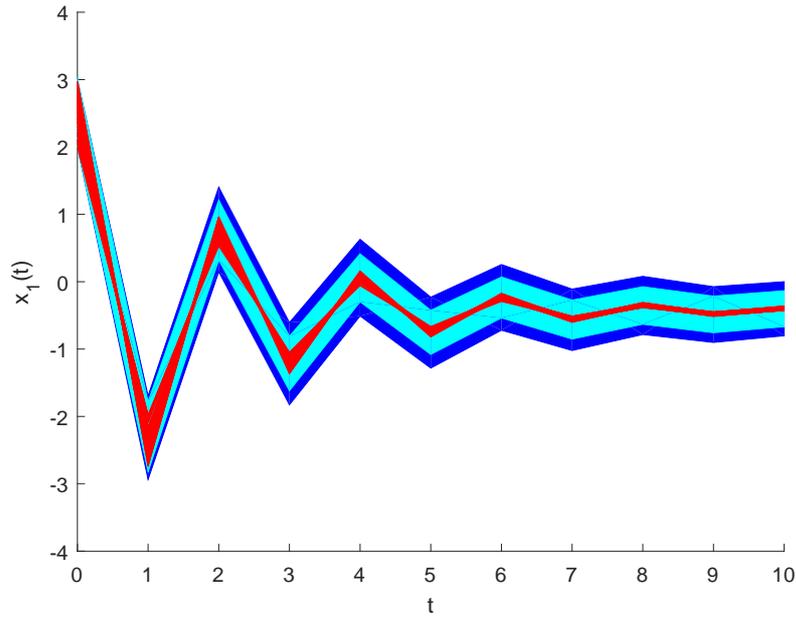}
		\caption{The reachable set estimations of state $\mathbf{x}_1(t)$ for the proposed system in time interval $[0,10]$ with partition numbers $M_i=5,20$. Blue area is for $M_i=5$ and cyan area is for $M_i=20$. The estimation result with $M_i=20$ is tighter than that by $M_i=5$.  $1000$ random state trajectories (red lines) are all within the estimated reachable sets.  }
		\label{fig3}
	\end{center}
\end{figure}

\begin{figure}
	\begin{center}
		\includegraphics[width=12cm]{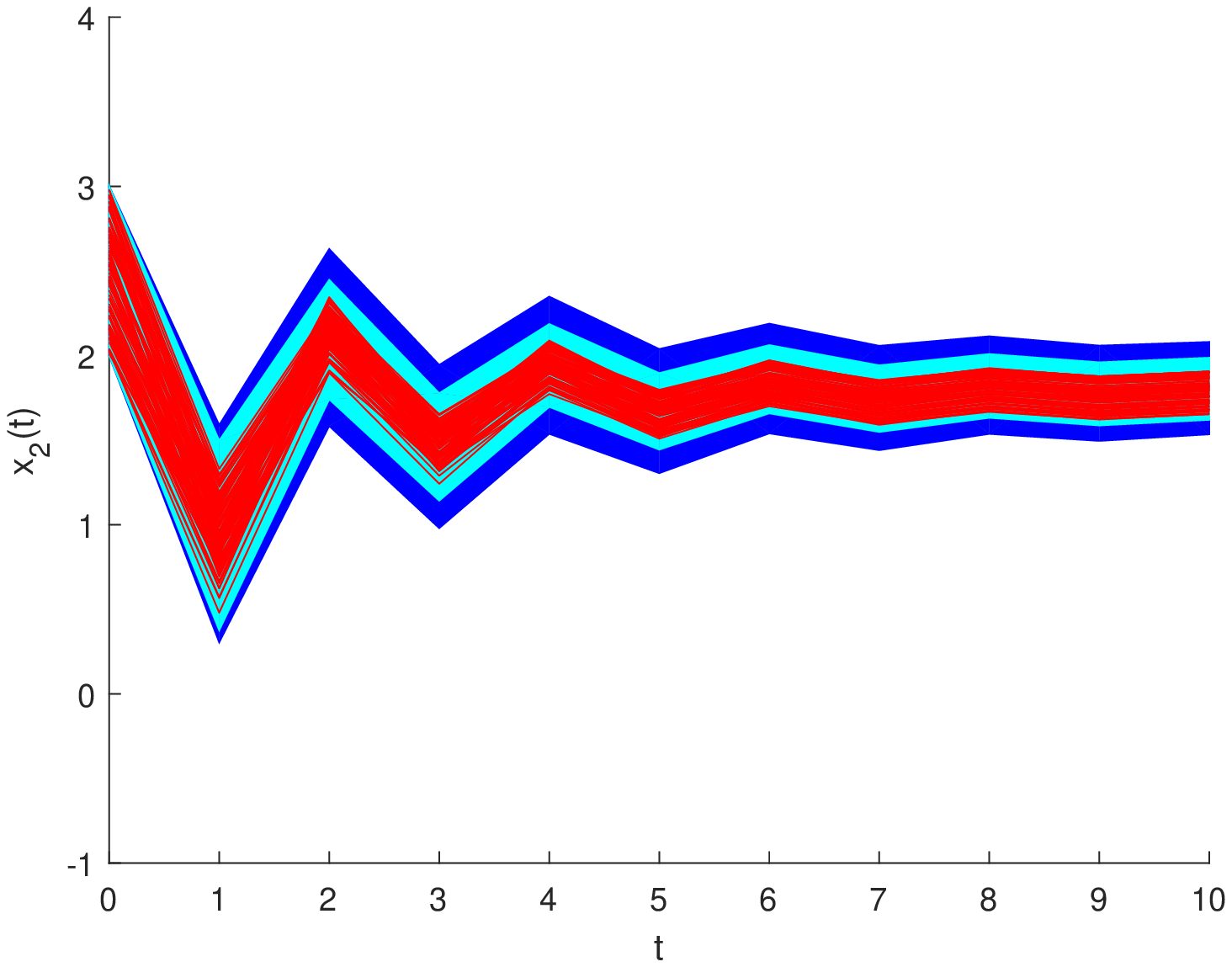}
		\caption{The reachable set estimations of state $\mathbf{x}_2(t)$ for the proposed system in time interval $[0,10]$ with partition numbers $M_i=5,20$. Blue area is for $M_i=5$ and cyan area is for $M_i=20$. The estimation result with $M_i=20$ is tighter than that by $M_i=5$.  $1000$ random state trajectories (red lines) are all within the estimated reachable sets. }
		\label{fig4}
	\end{center}
\end{figure}

\begin{figure}
	\begin{center}
		\includegraphics[width=12cm]{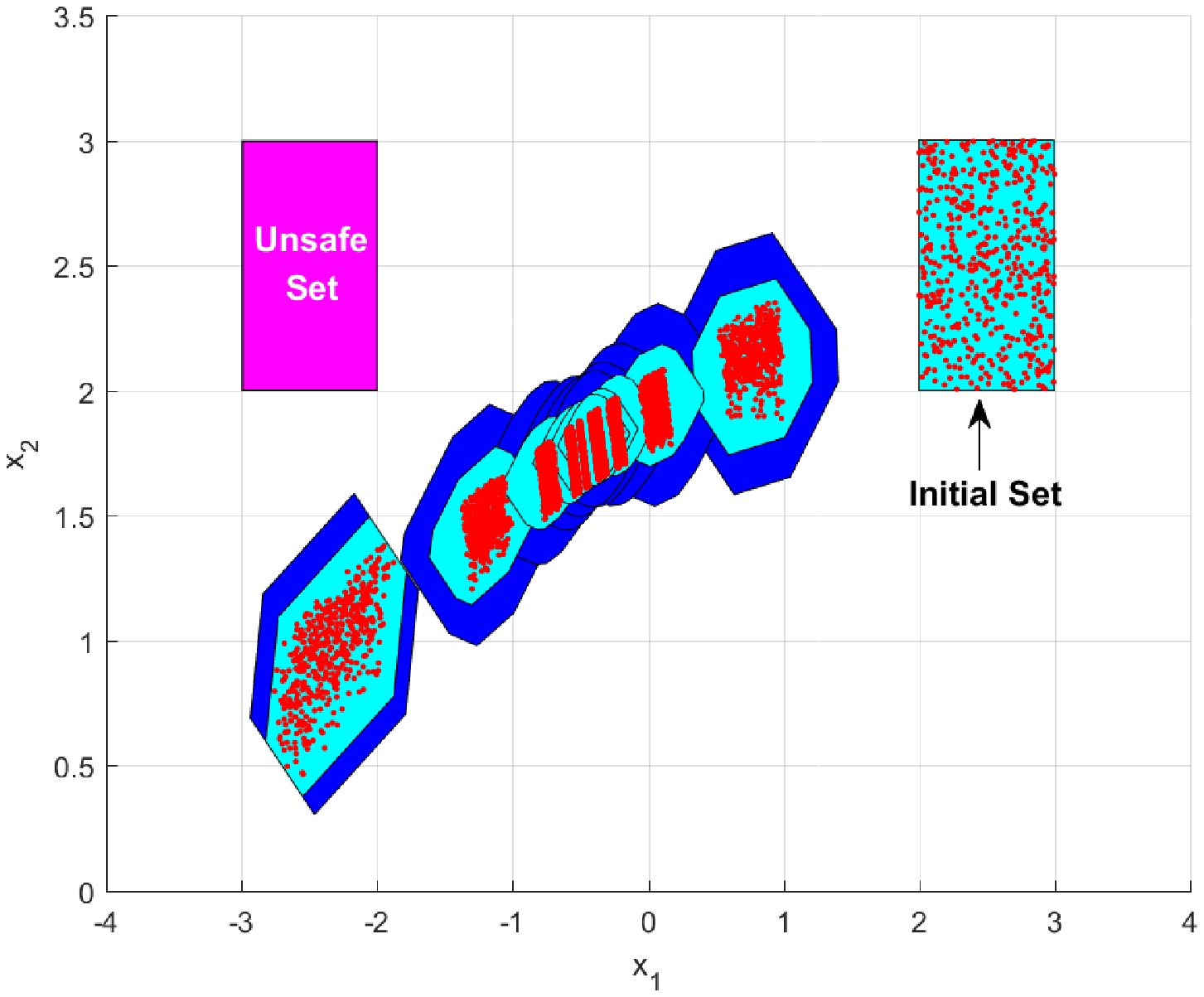}
		\caption{\boldmath The reachable set estimations of state $\mathbf{x}(t)$ for the proposed system in time interval $[0,10]$ with partition numbers $M_i=5,20$. Blue area is for $M_i=5$ and cyan area is for $M_i=20$. The estimation result with $M_i=20$ is tighter than that by $M_i=5$.  $1000$ random state trajectories (red spots) are all within the estimated reachable sets. There exists no intersection between the estimated reachable set and the proposed unsafe region, so the system can be claimed to be safe in interval $[0,10]$. } \label{fig2}
	\end{center}
\end{figure}

\section{Conclusion}
The problems of reachable set estimation and safety verification for a class of dynamical systems equipped with neural network controllers have been addressed in this paper. In most existing work, neural network components contained in closed-loop systems are considered as black-boxes lacking effective methods to predict all behaviors. This paper presents a novel approach to estimate the output reachable set of a feedforward neural networks known as multi-layer perceptrons (MLPs). By discretizing the input set into a number of hyper-rectangular cells, the output reachable set can be efficiently over-approximated via solving an LP problem. The proposed approach is not restricted to specific forms of activation functions of neurons. Combining the estimated output set of these neural network controllers and reachable set estimation methods for ODE models of plants, we develop algorithms to over-approximate the reachable set and verify safety property of closed-loop control systems with neural network controllers that may be implemented in embedded software. Beyond the initial results described in this paper, in the future we plan to investigate other modeling and reachability analysis approaches for the plant and neural network controllers, such as some of the state-of-the-art hybrid systems verification tools for consider the continuous-time plant reachability, like SpaceEx~\cite{frehse2011spaceex}, Flow*~\cite{chen2016rtss,chen2013flow}, Hyst~\cite{bak2015hscc}, Hylaa~\cite{bak2017hylaa}, C2E2~\cite{duggirala2015c2e2,fan2016emsoft,fan2016cav}, and others, as well as broader classes of neural networks.
\bibliographystyle{ieeetr}
\bibliography{ref}

\end{document}